\documentclass[11pt]{article}

\usepackage[margin=1in]{geometry}
\usepackage{amsmath,amssymb,amsthm}
\usepackage{authblk}
\usepackage{hyperref}
\usepackage{enumitem}

\newtheorem{theorem}{Theorem}
\newtheorem{lemma}{Lemma}
\newtheorem{corollary}{Corollary}
\newtheorem{remark}{Remark}

\DeclareMathOperator{\rank}{rank}
\DeclareMathOperator{\Ker}{Ker}

\newcommand{\R}{\mathbb{R}}

\newcommand{\II}{\mathrm{II}}
\newcommand{\proj}{\operatorname{proj}}

\title{Curvature, Minimality and Uniqueness of Equilibrium}

\author[1]{Andrea Loi}
\author[2]{Stefano Matta}
\affil[1]{Dipartimento di Matematica e Informatica, Universit\`a di Cagliari, Cagliari, Italy. Email: \texttt{andrea.loi@unica.it}.}
\affil[2]{Dipartimento di Scienze economiche e Aziendali, Universit\`a di Cagliari, Cagliari, Italy. Email: \texttt{stefano.matta@unica.it}.}

\date{}

\begin{document}

\maketitle

\begin{abstract}
For a smooth pure exchange economy with fixed aggregate resources, we study two
geometric conditions on the equilibrium manifold $E(r)$ endowed with the metric
induced from its Euclidean ambient space. First, for arbitrary numbers of
commodities and consumers, we prove that intrinsic flatness forces equilibrium
prices to be locally constant. Together with Balasko's uniqueness--constancy
criterion, this yields a necessary and sufficient condition: $E(r)$ is
intrinsically flat if and only if the normalized equilibrium price is unique for
every economy with aggregate resources $r$. This extends the
curvature--uniqueness theorem of \cite{LoiMatta2018} and completes the
higher-dimensional direction pursued in \cite{LoiMattaUccheddu2023}. Second, in
the two-commodity case, we show that minimality of $E(r)$ already forces local
constancy of the price map. Under the uniform-distribution interpretation of
\cite{LoiMatta2021}, this gives the minimal-entropy/uniqueness equivalence
without the additional asymptotic assumption used there. Both arguments rely on
the same local parametrization of $E(r)$ and avoid the explicit construction of a
normal frame.
\end{abstract}


\section{Introduction}

The equilibrium manifold approach reveals a deep connection between the geometry of the space of economies and the economic properties of equilibria. Since the seminal work of Debreu \cite{Debreu1970} and the global formulation developed by Balasko \cite{Balasko1988}, the equilibrium manifold has provided a natural geometric framework in which prices and endowments are treated jointly. In a smooth pure exchange economy, the equilibrium manifold is not merely a set of solutions of the market-clearing equations. It is a smooth object whose shape encodes how equilibrium prices, incomes and endowments co-vary when one moves across nearby economies.

The uniqueness of equilibrium price is central for comparative statics, stability and policy interpretation. The literature has developed several approaches to global uniqueness, ranging from classical sufficient conditions and surveys \cite{MasColell1991,Kehoe1998} to more recent restrictions based on preferences, endowments and demand shapes. In two-good economies, Geanakoplos and Walsh \cite{GeanakoplosWalsh2018} identify conditions involving DARA preferences, patience types and endowment restrictions; Toda and Walsh \cite{TodaWalsh2017} provide examples of multiplicity in Edgeworth-box economies. Other recent contributions use offer curves, individual demand shapes, and HARA/CRRA structures to obtain new uniqueness conditions \cite{Gimenez2022,Won2023,Won2024,LoiMatta2023,LoiMatta2024}; see also the recent overview by Toda and Walsh \cite{TodaWalsh2024}. The present note belongs to this broad literature, but it follows the geometric route: instead of imposing a specific functional form on preferences, it studies how the intrinsic geometry of the equilibrium manifold constrains the price map.

In \cite{LoiMatta2018}, we studied the Riemannian geometry of the fixed-resource equilibrium manifold $E(r)$ endowed with the metric induced from its Euclidean ambient space. We proved that, in the case of two commodities and an arbitrary number of consumers, zero curvature of $E(r)$ implies uniqueness of equilibrium. The restriction to two commodities was geometrically natural in that proof: the equilibrium manifold is then a hypersurface, so one can work with a unit normal vector field, the associated Gauss map, and explicit curvature formulae from classical submanifold geometry \cite{Carmo1976,Carmo1992}. The final discussion in \cite{LoiMatta2018} pointed out the main obstruction to extending the result to $L>2$: the equilibrium manifold is no longer a hypersurface, and one would have to deal with a frame of normal vectors or, equivalently, with a much more involved curvature tensor system.

A complementary extrinsic route was developed in \cite{LoiMatta2021}. There, uncertainty around an equilibrium is described by a probability model on $E(r)$ and Shannon's differential entropy is computed with respect to the Riemannian volume induced by the ambient Euclidean metric. Under the uniform distribution, minimizing entropy is equivalent to minimizing volume, so the economic question becomes whether the equilibrium manifold can be a minimal submanifold in the presence of price multiplicity. That paper proved the resulting minimal-entropy/uniqueness conjecture for an arbitrary number of commodities and two consumers, and for two commodities and an arbitrary number of consumers under an additional asymptotic condition on the Gauss map. One contribution of the present note is to remove that additional condition in the two-commodity case: minimality alone already forces the price map to be locally constant.

Related geometric work on the equilibrium manifold also includes \cite{LoiMatta2011}, where the choice of metric is connected with catastrophes. A further step in this geometric program was taken in \cite{LoiMattaUccheddu2023}. That paper proposed a geometric approach to equilibrium selection under changes in endowments and, as a by-product, proved the equivalence between zero curvature and uniqueness in the case of an arbitrary number of commodities and two consumers. Thus, the two previous results covered two complementary directions: \cite{LoiMatta2018} treated arbitrary consumers but only two commodities, whereas \cite{LoiMattaUccheddu2023} treated arbitrary commodities but only two consumers. The present note provides the missing higher-dimensional argument: arbitrary numbers of commodities and consumers.

The first purpose of this note is to show that the higher-codimension obstruction can be avoided by not computing the curvature explicitly. The key observation is local and structural: the standard parametrization of the fixed-resource equilibrium manifold is affine in the free endowment variables. This implies that certain pure endowment directions have vanishing second fundamental form. When the intrinsic curvature tensor vanishes, the Gauss equation forces the corresponding mixed terms of the second fundamental form to vanish as well. The novelty is that no normal frame and no explicit curvature tensor computation are needed.

Economically, these mixed terms are the relevant ones: they contain precisely the derivatives of the equilibrium price map. Thus intrinsic flatness imposes restrictions on the local price response to admissible redistributions of endowments. The remaining possibility, namely that certain canonical endowment directions become tangent to the equilibrium manifold, is ruled out by a purely linear-algebraic tangential lemma. The conclusion is that the Jacobian of the price map is zero. Hence intrinsically flat equilibrium manifolds carry locally constant prices.

The second purpose is to revisit the minimality approach of \cite{LoiMatta2021}. In the two-commodity case, the same local parametrization allows us to test the mean curvature equation against a carefully chosen normal direction. A simple parity argument in the free endowment variables shows that a minimal equilibrium manifold cannot have a nonconstant price map. Thus the entropy/minimal-volume route and the curvature route both lead to the same economic conclusion: price multiplicity is incompatible with the relevant geometric neutrality condition.

The contribution is therefore twofold. \emph{Technically}, the proof replaces global or asymptotic arguments based on a normal vector field with local arguments driven by the parametrization of $E(r)$. \emph{Economically}, it shows that both intrinsic flatness and, in the two-commodity case, extrinsic minimality have direct global content: they rule out local price variation and therefore yield uniqueness of the normalized equilibrium price for every economy with fixed aggregate resources.

The note is organized as follows. Section \ref{sec:setup} recalls the local parametrization of the equilibrium manifold with fixed resources. Section \ref{sec:lemma} proves the tangential lemma. Section \ref{sec:main} proves the intrinsic-flatness theorem and derives the curvature--uniqueness equivalence. Section \ref{sec:minimality} proves the corresponding minimality result in the two-commodity case and relates it to the minimal-entropy approach.

\section{The local parametrization of the equilibrium manifold}\label{sec:setup}

Consider a smooth pure exchange economy with $L$ commodities and $M$ consumers, with
\[
L\geq 2, \qquad M\geq 2.
\]
We work in the standard fixed-resource smooth exchange setting: preferences satisfy
the usual smoothness assumptions, prices are normalized, aggregate resources
$r\in\R^L$ are fixed, and the fixed-resource equilibrium manifold is the smooth
submanifold of the normalized price--endowment space described in Balasko's
framework. Let
\[
q=L-1, \qquad n=M-1.
\]
Prices are normalized by setting the last price equal to one, so that the normalized price vector may be written as
\[
p=(p_1,\ldots,p_q).
\]
Let $r\in \R^L$ denote fixed aggregate resources and let $E(r)$ be the corresponding equilibrium manifold.

We use the standard local description of $E(r)$ obtained from the price-income equilibrium manifold $B(r)$ and the free endowment coordinates. Since $B(r)$ has dimension $M-1=n$, take a local parametrization
\[
t=(t_1,\ldots,t_n)\longmapsto \bigl(p(t),w_1(t),\ldots,w_n(t)\bigr),
\]
where
\[
p(t)=\bigl(p_1(t),\ldots,p_q(t)\bigr)
\]
is the equilibrium price vector and $w_i(t)$ are the corresponding incomes of the first $n=M-1$ consumers. The last consumer's endowment is determined residually by the fixed-resource constraint.

For each $i=1,\ldots,n$, let
\[
\alpha_i=(\alpha_i^1,\ldots,\alpha_i^q)\in \R^q
\]
be the first $L-1$ endowment coordinates of consumer $i$. The last coordinate of consumer $i$ is then determined by the budget identity:
\[
z_i=w_i(t)-p(t)\cdot \alpha_i.
\]
Thus, after the last consumer's endowment has been eliminated by the resource constraint, we use the following local parametrization of $E(r)$:
\begin{equation}\label{eq:parametrization}
\begin{aligned}
F(t,\alpha)=\bigl(&p(t),\alpha_1,w_1(t)-p(t)\cdot\alpha_1,\ldots,\\
&\alpha_n,w_n(t)-p(t)\cdot\alpha_n\bigr).
\end{aligned}
\end{equation}
The ambient space has dimension
\[
(L-1)+L(M-1)=LM-1,
\]
while
\[
\dim E(r)=L(M-1).
\]
Hence the codimension is $L-1=q$.

For each endowment block $i$, denote by $e_{i,a}$ the canonical vector corresponding to the coordinate $\alpha_i^a$, $a=1,\ldots,q$, and by $e_{i,L}$ the canonical vector corresponding to the last coordinate $z_i$ of the same block. Differentiating \eqref{eq:parametrization} gives
\begin{equation}\label{eq:Falpha}
F_{\alpha_i^a}=e_{i,a}-p_a(t)e_{i,L},
\end{equation}
\begin{equation}\label{eq:mixed}
F_{t_k\alpha_i^a}=-\frac{\partial p_a}{\partial t_k}(t)e_{i,L},
\end{equation}
and
\begin{equation}\label{eq:alphaalpha}
F_{\alpha_i^a\alpha_i^a}=0.
\end{equation}
On the chosen coordinate neighbourhood, $F$ is an immersion. Thus the tangent
space $T_xE(r)$ at $x=F(t,\alpha)$ is generated by the vectors
\[
F_{t_1},\ldots,F_{t_n},\qquad F_{\alpha_i^a},\quad i=1,\ldots,n,\quad a=1,\ldots,q.
\]
These vectors form a local tangent frame.

\section{The tangential lemma}\label{sec:lemma}

The following elementary lemma is the linear-algebraic core of the proof.

\begin{lemma}[Tangential lemma]\label{lem:tangential}
Let $x=F(t,\alpha)\in E(r)$. If
\[
\proj_{N_xE(r)}(e_{i,L})=0
\]
for every $i=1,\ldots,n$, then
\[
Dp(t)=0.
\]
\end{lemma}

\begin{proof}
The assumption means that $e_{i,L}\in T_xE(r)$ for every $i=1,\ldots,n$. Therefore, for each $i$, there exist coefficients $\lambda_k^i$ and $\mu_j^{a,i}$ such that
\begin{equation}\label{eq:decomposition}
e_{i,L}=\sum_{k=1}^n \lambda_k^iF_{t_k}+\sum_{j=1}^n\sum_{a=1}^q \mu_j^{a,i}F_{\alpha_j^a}.
\end{equation}
We compare coordinates.

First, compare the price coordinates. The vector $e_{i,L}$ has zero price coordinates. The vectors $F_{\alpha_j^a}$ also have zero price coordinates, while the price coordinates of $F_{t_k}$ are $\partial p/\partial t_k$. Hence
\[
0=\sum_{k=1}^n \lambda_k^i\frac{\partial p}{\partial t_k}.
\]
Equivalently,
\begin{equation}\label{eq:kerDp}
Dp(t)\lambda^i=0,
\end{equation}
where $\lambda^i=(\lambda_1^i,\ldots,\lambda_n^i)^T$. Thus $\lambda^i\in \Ker Dp(t)$.

Second, compare the first $q$ coordinates of each endowment block. The vector $e_{i,L}$ has zero component in all these coordinates. The vectors $F_{t_k}$ also have zero component there. Moreover, among the vectors $F_{\alpha_j^a}$, the only one with a nonzero component in the coordinate $\alpha_j^a$ is $F_{\alpha_j^a}$ itself, and that component is one. Therefore all the coefficients $\mu_j^{a,i}$ vanish. Consequently \eqref{eq:decomposition} reduces to
\begin{equation}\label{eq:simplified}
e_{i,L}=\sum_{k=1}^n \lambda_k^iF_{t_k}.
\end{equation}

Third, compare the last coordinates of the endowment blocks. Define the $n\times n$ matrix
\begin{equation}\label{eq:Balpha}
B_\alpha(t)=\left(\frac{\partial w_m}{\partial t_k}-\frac{\partial p}{\partial t_k}\cdot \alpha_m\right)_{m,k=1}^n.
\end{equation}
The $k$-th column of $B_\alpha(t)$ is precisely the vector of last-block coordinates of $F_{t_k}$. Since the last-block coordinates of $e_{i,L}$ form the $i$-th canonical vector $e_i\in \R^n$, \eqref{eq:simplified} gives
\begin{equation}\label{eq:B-lambda}
B_\alpha(t)\lambda^i=e_i.
\end{equation}
By \eqref{eq:kerDp}, each $\lambda^i$ belongs to $\Ker Dp(t)$. Hence every canonical vector $e_i$ of $\R^n$ belongs to the image of the restricted map
\[
B_\alpha(t)|_{\Ker Dp(t)}:\Ker Dp(t)\longrightarrow \R^n.
\]
Therefore this restricted map is surjective. Hence
\[
\rank\bigl(B_\alpha(t)|_{\Ker Dp(t)}\bigr)=n.
\]
Since the rank of a linear map cannot exceed the dimension of its domain,
\[
n\leq \dim \Ker Dp(t).
\]
But $Dp(t):\R^n\to \R^q$, so $\Ker Dp(t)\subseteq \R^n$ and therefore
\[
\dim \Ker Dp(t)\leq n.
\]
It follows that
\[
\dim \Ker Dp(t)=n.
\]
Thus $\Ker Dp(t)=\R^n$, and consequently $Dp(t)=0$. In other words, if all residual
last-good directions are tangent, the price-income directions must generate all
last-block variations through vectors lying in $\Ker Dp(t)$, which is possible
only when $\Ker Dp(t)=\R^n$.
\end{proof}

\section{Intrinsic flatness and constant prices}\label{sec:main}

We now prove the main result. Let $E(r)$ be endowed with the Riemannian metric induced by the Euclidean metric of the ambient space. We say that $E(r)$ is intrinsically flat when its Riemann curvature tensor vanishes identically.

\begin{theorem}\label{thm:main}
Assume that the equilibrium manifold $E(r)$ is endowed with the induced Euclidean metric and that
\[
R^{E(r)}\equiv 0.
\]
Then the equilibrium price map $p$ is locally constant.
\end{theorem}

\begin{proof}
Since the ambient space is Euclidean, the second fundamental form is obtained by projecting second derivatives onto the normal space:
\begin{equation}\label{eq:II-projection}
\II(F_u,F_v)=\proj_{N_xE(r)}(F_{uv}).
\end{equation}
By \eqref{eq:alphaalpha},
\[
\II(F_{\alpha_i^a},F_{\alpha_i^a})=0.
\]
Moreover, by \eqref{eq:mixed} and \eqref{eq:II-projection},
\begin{equation}\label{eq:II-mixed}
\II(F_{t_k},F_{\alpha_i^a})=-\frac{\partial p_a}{\partial t_k}\proj_{N_xE(r)}(e_{i,L}).
\end{equation}

Apply the Gauss equation with
\[
X=F_{\alpha_i^a},\qquad Y=F_{t_k}.
\]
Since the ambient curvature is zero,
\begin{equation}\label{eq:gauss}
\langle R^{E(r)}(X,Y)Y,X\rangle
=
\langle \II(X,X),\II(Y,Y)\rangle-\|\II(X,Y)\|^2.
\end{equation}
Thus, in the present Euclidean ambient space, the intrinsic curvature of $E(r)$ is
entirely determined by the second fundamental form.
By assumption, $R^{E(r)}\equiv 0$, and by the previous paragraph $\II(X,X)=0$. Hence \eqref{eq:gauss} gives
\[
\|\II(F_{\alpha_i^a},F_{t_k})\|^2=0.
\]
Thus
\[
\II(F_{\alpha_i^a},F_{t_k})=0.
\]
Using \eqref{eq:II-mixed}, we obtain
\begin{equation}\label{eq:master}
\frac{\partial p_a}{\partial t_k}\proj_{N_xE(r)}(e_{i,L})=0
\end{equation}
for every $i=1,\ldots,n$, $a=1,\ldots,q$ and $k=1,\ldots,n$.

Suppose, by contradiction, that $Dp(t)\neq 0$. Then there exist $a_0$ and $k_0$ such that
\[
\frac{\partial p_{a_0}}{\partial t_{k_0}}\neq 0.
\]
Taking $a=a_0$ and $k=k_0$ in \eqref{eq:master}, we get
\[
\proj_{N_xE(r)}(e_{i,L})=0
\]
for every $i=1,\ldots,n$. By Lemma \ref{lem:tangential}, this implies $Dp(t)=0$, contradicting the assumption. Therefore $Dp(t)=0$.
Since this holds in local coordinates, the price map $p$ is locally constant.
\end{proof}

We shall use Balasko's result that, in the fixed-resource setting, uniqueness of
the normalized equilibrium price for every economy implies that the equilibrium
price does not depend on the endowment profile. Conversely, if the price
projection is constant on $E(r)$, uniqueness of the normalized equilibrium price
is immediate. Together, these observations give the uniqueness--constancy
criterion used below.

\begin{corollary}
In the standard fixed-resource smooth exchange setting, the following conditions
are equivalent:
\begin{enumerate}
\item the equilibrium manifold \(E(r)\), endowed with the metric induced by the
Euclidean ambient space, is intrinsically flat, i.e. \(R^{E(r)}\equiv 0\);
\item the normalized equilibrium price associated with each economy
\(\omega\in\Omega(r)\) with aggregate resources \(r\) is unique.
\end{enumerate}
\end{corollary}

\begin{proof}
Assume first that \(R^{E(r)}\equiv 0\). By Theorem~\ref{thm:main}, the equilibrium price map is
locally constant in every local parametrization of \(E(r)\). By the standard global parametrization of the fixed-resource equilibrium manifold,
\(E(r)\) is diffeomorphic to a Euclidean space and hence connected. Hence the price map is constant on \(E(r)\). Let
\(\omega\in\Omega(r)\), and suppose that \(p\) and \(\widetilde p\) are two
normalized equilibrium prices associated with \(\omega\). Then
\[
(p,\omega)\in E(r),
\qquad
(\widetilde p,\omega)\in E(r).
\]
Since the price map is constant on \(E(r)\), both prices must coincide with the
same constant vector. Therefore \(p=\widetilde p\), and the
normalized equilibrium price is unique for every economy with fixed total resources \(r\).

\noindent Conversely, suppose that the normalized equilibrium price is unique for every
\(\omega\in\Omega(r)\) with aggregate resources \(r\). By Balasko's theorem \cite[Theorem 7.3.9]{Balasko1988}, the equilibrium price associated
with \(\omega\) does not depend on \(\omega\). Thus the equilibrium manifold has
constant price component and is an affine submanifold parallel to the
endowment space. Its induced Riemannian curvature is therefore zero.
\end{proof}

\begin{remark}
The proof of Theorem~\ref{thm:main} does not require the explicit construction
of a normal frame. This is the key difference with the codimension-one
Gauss-map argument used in \cite{LoiMatta2018}. In higher codimension, the
second fundamental form is normal-vector-valued, and an explicit curvature
tensor computation would be considerably less tractable. The argument above
uses only those components of the second fundamental form that are forced by the
affine dependence of the parametrization on the endowment variables, together
with the Gauss equation and the tangential lemma.
\end{remark}

\begin{remark}
The mixed terms \(\II(F_{t_k},F_{\alpha_i^a})\) that intrinsic flatness forces
to vanish admit a direct economic reading. They couple a variation along the
price-income manifold \(B(r)\) with a redistribution of endowments, and by
\eqref{eq:II-mixed} their projection onto the normal bundle scales with the
price derivatives \(\partial p_a/\partial t_k\). Their vanishing is therefore
the geometric mechanism through which intrinsic flatness rules out local price
responses to admissible redistributions of endowments, which is the content of
the constancy condition underlying Balasko's uniqueness criterion.
\end{remark}

\section{Minimality and constant prices when \texorpdfstring{$L=2$}{L=2}}\label{sec:minimality}

We now record a complementary extrinsic result. In \cite{LoiMatta2021}, under a
uniform distribution on a neighborhood of the equilibrium manifold, the minimal
entropy property is equivalent to a minimal volume property. Geometrically, this
leads to the requirement that $E(r)$ be a minimal submanifold of its Euclidean
ambient space. The result below shows that, when $L=2$, minimality alone already
forces the price map to be locally constant.

For this section only, assume $L=2$. Then $q=L-1=1$, and the normalized price is
a scalar function $p=p(t)$, with $t\in\R^n$ and $n=M-1$. The local parametrization
becomes
\begin{equation}\label{eq:minimal-param}
F(t,\alpha)=
\bigl(p(t),\alpha_1,w_1(t)-p(t)\alpha_1,\ldots,
\alpha_n,w_n(t)-p(t)\alpha_n\bigr).
\end{equation}
Recall that an immersion is minimal if its mean curvature vector vanishes
identically. Equivalently, for every normal vector $N$ one has
\begin{equation}\label{eq:minimal-equation}
\sum_{A,B}g^{AB}\langle F_{AB},N\rangle=0,
\end{equation}
where $(g^{AB})$ is the inverse of the induced metric matrix.

\begin{theorem}\label{thm:minimal}
Assume \(L=2\). If the immersion \eqref{eq:minimal-param} is minimal with respect
to the Euclidean metric induced from the ambient space, then the equilibrium
price map \(p\) is locally constant.
\end{theorem}

\begin{proof}
Suppose, by contradiction, that \(Dp(t_0)\neq 0\) at some point \(t_0\).
Since \(p\) is scalar, the implicit function theorem allows us, after
restricting to a sufficiently small neighbourhood of \(t_0\), to use
\[
x=p(t)
\]
as one of the local coordinates. We write the remaining coordinates as
\[
s=(s_1,\ldots,s_{n-1}),
\]
so that \(t=(x,s)\). Put
\[
w_i(t)=W_i(x,s),
\qquad i=1,\ldots,n.
\]
Then the local parametrization becomes
\[
F(x,s,\alpha)=
\bigl(x,\alpha_1,W_1(x,s)-x\alpha_1,\ldots,
\alpha_n,W_n(x,s)-x\alpha_n\bigr).
\]
Set
\[
\beta_i=W_{i,x}-\alpha_i,
\qquad
\beta=(\beta_1,\ldots,\beta_n).
\]
The tangent vectors are
\[
F_x=(1,0,\beta),
\]
\[
F_{s_r}=(0,0,W_{s_r}),
\qquad r=1,\ldots,n-1,
\]
and
\[
F_{\alpha_i}=(0,e_i,-x e_i),
\qquad i=1,\ldots,n.
\]
Here \(e_i\) denotes the \(i\)-th canonical vector of \(\mathbb R^n\), and
\[
W_{s_r}
=
\left(
\frac{\partial W_1}{\partial s_r},
\ldots,
\frac{\partial W_n}{\partial s_r}
\right).
\]

Since \((x,s)\) are local coordinates on the price-income manifold, the tangent
vectors
\[
F_{s_1},\ldots,F_{s_{n-1}}
\]
are linearly independent. As
\[
F_{s_r}=(0,0,W_{s_r}),
\]
it follows that
\[
W_{s_1},\ldots,W_{s_{n-1}}
\]
are linearly independent in \(\mathbb R^n\). Hence their orthogonal complement
in \(\mathbb R^n\) is one-dimensional. Choose a nonzero vector
\(y\in\mathbb R^n\) such that
\[
y\cdot W_{s_r}=0,
\qquad r=1,\ldots,n-1.
\]
When \(n=1\), this condition is void.

Define
\[
N=(-\beta\cdot y,xy,y).
\]
Then \(N\) is normal to the immersion. Indeed,
\[
\langle N,F_{\alpha_i}\rangle
=
xy_i-xy_i=0,
\]
\[
\langle N,F_{s_r}\rangle
=
y\cdot W_{s_r}=0,
\]
and
\[
\langle N,F_x\rangle
=
-\beta\cdot y+\beta\cdot y=0.
\]

The endowment variables \(\alpha_i\) are free local coordinates. Hence, after
possibly restricting the coordinate neighbourhood, we may vary them along the
line
\[
\alpha_i=W_{i,x}-\tau y_i,
\qquad i=1,\ldots,n,
\]
for all \(\tau\) in some nonempty open interval \(I\subset\mathbb R\).
Equivalently, along this line one has
\[
\beta=\tau y.
\]
At these points,
\[
N_\tau=(-\tau |y|^2,xy,y).
\]
Set
\[
Y=|y|^2>0.
\]
Introduce the endowment direction
\[
\partial_\eta=\sum_{i=1}^n y_i\partial_{\alpha_i}.
\]
Then
\[
F_\eta
=
\sum_{i=1}^n y_iF_{\alpha_i}
=
(0,y,-xy).
\]

We now compute the part of the induced metric on the two-dimensional subspace
generated by \(F_x\) and \(F_\eta\). Since \(\beta=\tau y\), we have
\[
F_x=(1,0,\tau y).
\]
Thus
\[
\langle F_x,F_x\rangle=1+\tau^2Y,
\]
\[
\langle F_x,F_\eta\rangle=-x\tau Y,
\]
and
\[
\langle F_\eta,F_\eta\rangle=(1+x^2)Y.
\]
Hence the corresponding \(2\times 2\) block is
\[
G_0(\tau)=
\begin{pmatrix}
1+\tau^2Y & -x\tau Y\\
-x\tau Y & (1+x^2)Y
\end{pmatrix}.
\]
Its determinant is
\[
\det G_0(\tau)=Y(1+x^2+\tau^2Y).
\]

Since the mean curvature vector is the trace of the second fundamental form, the
scalar equation
\[
\langle H,N_\tau\rangle=0
\]
may be evaluated in any tangent basis at the point under consideration. We shall
therefore replace the coordinate basis in the \(\alpha\)-directions by a basis
adapted to the orthogonal decomposition
\[
\mathbb R^n=\mathbb R y\oplus y^\perp.
\]
Moreover, since \(N_\tau\) is normal, if \(E_A,E_B\) are locally extended as
tangent vector fields, then
\[
\langle \II(E_A,E_B),N_\tau\rangle
=
\langle D_{E_A}E_B,N_\tau\rangle,
\]
where \(D\) is the Euclidean connection. In what follows, the fields
\(F_\eta\) and \(F_{v_\ell}\), \(\ell=2,\ldots,n\), are obtained from the fields
\(F_{\alpha_i}\) by taking constant linear combinations at the point under
consideration. Therefore the Euclidean derivatives displayed below compute the
corresponding second-fundamental-form coefficients after pairing with
\(N_\tau\).

Choose an orthonormal basis
\[
v_2,\ldots,v_n
\]
of \(y^\perp\). For \(\ell=2,\ldots,n\), set
\[
F_{v_\ell}
=
\sum_{i=1}^n (v_\ell)_i F_{\alpha_i}
=
(0,v_\ell,-xv_\ell).
\]
Since \(v_\ell\cdot y=0\), one has
\[
\langle F_x,F_{v_\ell}\rangle
=
-x\tau\, y\cdot v_\ell
=
0,
\]
and
\[
\langle F_\eta,F_{v_\ell}\rangle
=
(1+x^2)y\cdot v_\ell
=
0.
\]
Moreover, since \(y\cdot W_{s_r}=0\),
\[
\langle F_x,F_{s_r}\rangle
=
\tau y\cdot W_{s_r}
=
0,
\]
and
\[
\langle F_\eta,F_{s_r}\rangle
=
-x y\cdot W_{s_r}
=
0.
\]
Thus, at the point under consideration and along the line \(\beta=\tau y\), the
plane generated by \(F_x\) and \(F_\eta\) is orthogonal to the directions
generated by \(F_{s_r}\) and by \(F_{v_\ell}\), \(\ell=2,\ldots,n\). Therefore
the metric is block diagonal with respect to the splitting
\[
\operatorname{span}\{F_x,F_\eta\}
\oplus
\operatorname{span}\{F_{s_1},\ldots,F_{s_{n-1}},F_{v_2},\ldots,F_{v_n}\}.
\]
Consequently the inverse metric also splits. In particular, the coefficient of
the full inverse metric corresponding to the pair \((x,\eta)\) is the mixed
coefficient of \(G_0(\tau)^{-1}\), namely
\[
g^{x\eta}
=
\frac{x\tau}{1+x^2+\tau^2Y}.
\]

Next,
\[
F_{x\eta}
=
\sum_{i=1}^n y_iF_{x\alpha_i}
=
-\sum_{i=1}^n y_i e_{i,L}.
\]
Hence
\[
\langle \II(F_x,F_\eta),N_\tau\rangle
=
\langle F_{x\eta},N_\tau\rangle
=
-\sum_{i=1}^n y_i^2
=
-Y.
\]
Thus the mixed contribution of the pair \((x,\eta)\) to the scalar minimality
equation is
\[
2g^{x\eta}\langle \II(F_x,F_\eta),N_\tau\rangle
=
-\frac{2x\tau Y}{1+x^2+\tau^2Y}.
\]

We now check that no other term contributes an odd part in \(\tau\). Since the
immersion is affine in the endowment variables,
\[
D_{F_\eta}F_\eta=0,
\qquad
D_{F_\eta}F_{v_\ell}=0,
\qquad
D_{F_{v_\ell}}F_{v_m}=0
\]
for all \(\ell,m=2,\ldots,n\). Thus second fundamental form terms involving only
endowment directions give no contribution after pairing with \(N_\tau\).

Moreover,
\[
D_{F_x}F_{v_\ell}
=
-\sum_{i=1}^n (v_\ell)_i e_{i,L}.
\]
Therefore
\[
\langle \II(F_x,F_{v_\ell}),N_\tau\rangle
=
\left\langle D_{F_x}F_{v_\ell},N_\tau\right\rangle
=
-\sum_{i=1}^n (v_\ell)_i y_i
=
-v_\ell\cdot y
=
0
\]
for every \(\ell=2,\ldots,n\). Hence all mixed terms involving \(F_x\) and a
transverse endowment direction \(F_{v_\ell}\) vanish.

Similarly,
\[
D_{F_{s_r}}F_\eta=0,
\qquad
D_{F_{s_r}}F_{v_\ell}=0,
\]
because
\[
F_{s_r}=(0,0,W_{s_r})
\]
is independent of the endowment variables.

It remains to consider the second fundamental form terms involving only the
directions \(F_x,F_{s_1},\ldots,F_{s_{n-1}}\), namely the terms coming from
\[
F_{xx},\qquad F_{xs_r},\qquad F_{s_rs_l}.
\]
Their first, that is price, component is zero, because \(x\) itself is the first
coordinate of the parametrization. Since
\[
N_\tau=(-\tau Y,xy,y)
\]
depends on \(\tau\) only through its first component, it follows that
\[
\langle F_{xx},N_\tau\rangle,
\qquad
\langle F_{xs_r},N_\tau\rangle,
\qquad
\langle F_{s_rs_l},N_\tau\rangle
\]
are independent of \(\tau\).

Their inverse metric coefficients belong either to the \((x,\eta)\)-block or to
the complementary block. In the \((x,\eta)\)-block, the diagonal coefficients of
\(G_0(\tau)^{-1}\) are even functions of \(\tau\). The complementary block is
independent of \(\tau\), because it is generated by \(F_{s_r}\) and by
\(F_{v_\ell}\), none of which depends on \(\tau\). Therefore these remaining
terms are even functions of \(\tau\).

Consequently, in the scalar minimality equation
\[
\langle H,N_\tau\rangle
=
\sum_{A,B}g^{AB}
\langle \II(E_A,E_B),N_\tau\rangle
=
0,
\]
computed in the adapted tangent basis
\[
E_A\in
\{F_x,F_\eta,F_{s_1},\ldots,F_{s_{n-1}},F_{v_2},\ldots,F_{v_n}\},
\]
the odd part is precisely
\[
-\frac{2x\tau Y}{1+x^2+\tau^2Y}.
\]
Since the immersion is minimal, this equation holds for every \(\tau\in I\).
Moreover, the left-hand side is a rational function of \(\tau\). Hence it
vanishes identically as a rational function, and therefore its odd part must
vanish identically. Thus
\[
-\frac{2x\tau Y}{1+x^2+\tau^2Y}=0
\qquad
\text{for every }\tau.
\]
Because \(Y>0\), this forces
\[
x=0.
\]
But \(x=p(t)\) is a normalized equilibrium price, and prices are strictly
positive. This contradiction proves that no point \(t_0\) with
\(Dp(t_0)\neq 0\) can exist. Hence
\[
Dp=0
\]
on the coordinate neighbourhood. Therefore the equilibrium price map is locally
constant.
\end{proof}

\begin{remark}
Theorem~\ref{thm:minimal} strengthens the two-commodity part of
\cite{LoiMatta2021}. The earlier argument required an additional assumption on
the Gauss map outside a compact set. The proof above is purely local: it uses
only the standard parametrization of \(E(r)\), the freedom of the endowment
coordinates, and the positivity of normalized prices.
\end{remark}

\end{document}